\documentclass{article}
\usepackage[utf8]{inputenc}

\usepackage{geometry}
\usepackage[pagebackref,colorlinks=true,urlcolor=red,linkcolor=red,citecolor=red]{hyperref}

\usepackage{epsfig,graphics, graphicx}
\usepackage{subcaption, comment}
\usepackage[percent]{overpic}
\usepackage{hyperref}

\usepackage{amssymb,bm}
\usepackage{amsthm}
\usepackage{color}
\usepackage[]{amsmath}
\usepackage[]{amsfonts}
\usepackage[]{fancyhdr}
\usepackage[]{graphicx,wrapfig}
\graphicspath{{/EPSF/}{../figures/}{figures/}}

\usepackage{graphicx}
\usepackage[dvipsnames]{xcolor}
\usepackage{verbatim}
\usepackage{pbox}

\usepackage{tcolorbox}

\usepackage{xypic}
\usepackage{enumerate}

\usepackage{tikz}
\usetikzlibrary{calc}

\usepackage{placeins}

\newtheorem{theorem}{Theorem}
\newtheorem{definition}{Definition}
\newtheorem{remark}{Remark}

\newtheorem{lemma}{Lemma}
\newtheorem{proposition}{Proposition}

\newcommand{\calA}{\mathcal{A}}
\newcommand{\zplus}{\mathbb{Z}^+}

\newcommand{\prodiS}{\prod_{i \in \partial S} }

\title{ Lattice models and super telescoping formula}
\author{Mohammad Javad Latifi Jebelli\footnote{Email: mohammad.javad.latifi.jebelli@dartmouth.edu} \\ Dartmouth College}
\date{January 2023}

\begin{document}

\maketitle

 \begin{abstract}

  In this paper, we introduce the super telescoping formula, a natural generalization of well-known telescoping formula.  We  explore various aspects of the formula including its origin and the telescoping cancellations emerging from symmetric patterns. We also show that the super telescoping formula leads to the construction of exactly solvable lattice models with interesting partition functions.
 \end{abstract}

\section{Super Telescoping Formula}\label{sec:suptel}

 Let $S\subset \mathbb{R}$ be a finite union of intervals of the form $[a,b]$ with $a,b \in \mathbb{Z}^+=\{ 0,1,2,\dots\}$. If $S = [a_1, b_1] \cup \dots \cup [a_L,b_L]$ the boundary of $S$ is denoted by $\partial S = \{a_1, b_1, \dots , a_L, b_L \} \subset \mathbb{Z}^+$. Let us associate a number (weight) to such integer subset, 

 $$
 w(S) = \prod_{i\in \partial S} \frac{1}{i+1} = \prod_{j=1}^L \frac{1}{(a_j +1)(b_j +1)}
 $$

As in figure \ref{fig:int_subset_024568}, It is beneficial to picture an integer subset $S$ as a collection of line segments on the real line. Next, we are going to sum $w(S)$ over certain family of integer subsets. In the case where this integer subsets are single line segments of fixed length, this gives the well known telescoping formula from calculus. More specifically think of a line segment of fixed length $1$ moving in the positive direction of the real line. In other words, let $S_k = [k-1,k],\, k=1,2,3,\dots$ and hence $w(S_k) = \frac{1}{k(k+1)}$. The telescoping cancellation comes from the identity $w(S_k) = \frac{1}{k} - \frac{1}{k+1}$  , and implies that $\sum_{k=1}^{\infty} w(S_k) = 1$. 

\begin{figure}[!h]
    \centering
    \includegraphics[width=300px]{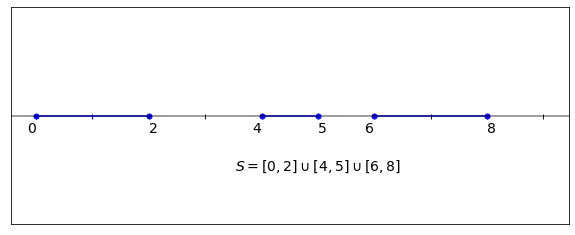}
    \caption{An example of integer subset $S\in \mathcal{A}_5$ with $w(S) = \frac{1}{(1)(3)} \cdot \frac{1}{(5)(6)} \cdot \frac{1}{(7)(9)}$.}
    \label{fig:int_subset_024568}
\end{figure}


Now, the question is what happens in a more general case? Fix some $n\in \mathbb{N}$ and consider the family $\mathcal{A}_n$ of all integer subsets with total length of $n$. That means  $S = [a_1, b_1] \cup \dots \cup [a_L,b_L]$ belongs to $\mathcal{A}_n$ if and only if $|S| = \sum_{j=1}^L (b_j - a_j) = n$. Figure \ref{fig:int_subset_size3} demonstrates some examples of sets in $\mathcal{A}_3$. As a surprsing special case, the super telescoping formula implies $\sum_{S\in \calA_n} w(S) = 1$, independent from the choice of $n$.

$$
\sum_{S\in \calA_n} \left( \prod_{i\in \partial S} \frac{1}{i+1} \right) = 1
$$

\begin{figure}[!h]
    \centering
    \includegraphics[width=250px]{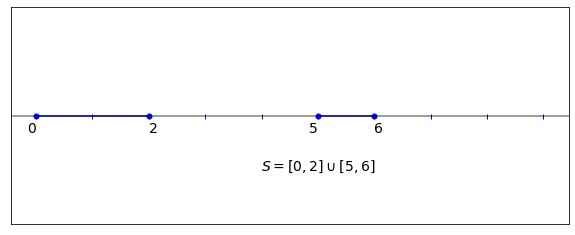}
    \includegraphics[width=250px]{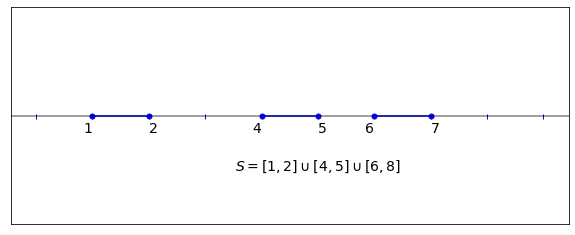}
    \includegraphics[width=250px]{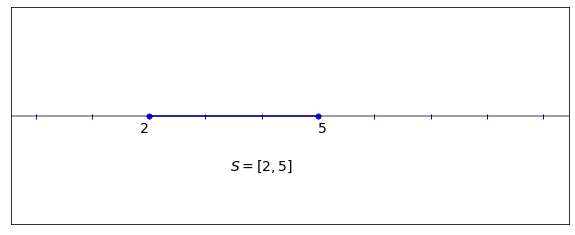}

    \caption{Examples of integer subsets in $\mathcal{A}_3$. These sets have total length $3$ and can be positioned anywhere on the real line. }
    \label{fig:int_subset_size3}
\end{figure}


 We study the following generalization of the above sum. For arbitrary non-zero parameter $\beta$, one has (see proof of theorem \ref{thm:supertel})

\begin{equation}\label{eq:var-formula}
\sum_{S\in \mathcal{A}_n} \prodiS \frac{1}{i\beta +1 } = \left(\begin{array}{c} \beta^{-1} + n -1 \\ n \end{array}\right)
\end{equation}

It is easy to see that for $\beta=1$ the right hand side is constant 1, independent of $n$.  Comparing elementary telescoping sums, the cancellation patterns here are more complex and we will discuss them in section \ref{sec:tel-cancellation}. These cancellation patterns involve an interesting family of symmetries of subsets of integers. Despite the fact that the left-hand side of formula \ref{eq:var-formula} has poles at $\beta = -\frac{1}{i}$ for every integer $i$, these poles are removable. We provide explicit proof of the following theorem using cancellations.


\begin{theorem} \label{thm:poles}
In terms of the variable $\beta$, all of the poles appearing on the left-hand side of formula \ref{eq:var-formula} are removable. 
\end{theorem} 

 In addition to the above statement, in section \ref{sec:tel-cancellation} we prove explicitly that the left-hand side of formula \ref{eq:var-formula} reduces to a polynomial of degree $n$ in $\beta^{-1}$, consistent with the right-hand side. However, establishing explicit and direct proof for formula \ref{eq:var-formula} or the following super-telescoping formula remains as open problem (despite the simplicity of the expression, the proof seems to be surprisingly challenging when $\beta \neq 1$).  Here we give an indirect proof based on results from random matrix theory and the ideas in \cite{Latifi-Pickrell}, \\



\begin{theorem}\label{thm:supertel}
(\textbf{Super telescoping formula}) Let $\mathcal{A}=\mathcal{A}_0 \cup \mathcal{A}_1 \cup \mathcal{A}_2 \cup \dots$ denote the family of integer subsets of $\mathbb{Z}^+$ with arbitrary size, then  
\begin{equation}\label{eq:super_tel}
     \sum_{S\in \calA} \left( \prodiS \frac{1}{i\beta +1 } \right) x^{|S|} = \frac{1}{(1-x)^{1/\beta}}
\end{equation}
\end{theorem}
\begin{proof} 
Following the result of Chhaibi and Najnudel \cite{Chhaibi},  we have a correspondence of probability measures  (see Appendix I for details)

$$ \prod_{n=1}^{\infty}\frac{n\beta}{\pi}e^{-n\beta\vert
f_n\vert^2}d\lambda(f_n) \longleftrightarrow \prod_{n=1}^{\infty}
\frac{n\beta}{\pi}(1-\vert\alpha_{n}\vert^2)^{n\beta-1}d\lambda(\alpha_{n}) $$
here the trace of a random function $f(z) = f_-(z) + f_0 + f_+(z) = \sum f_n z^n$ on the boundary of the unit disk gives a random measure on $S^1$, leading to a random Verblunsky sequence $\alpha_n$. If we define $x_n$ (in terms of $f_n$) by 
$$
e^{-f_+(z)}=1+\sum_{1}^{\infty} x_{n} z^{n}
$$
and $ \hat{x}_n$ (in terms of $\alpha_n$) by 

\begin{equation} \label{eq:def_xhat}
    r(z) = \gamma(z)+\delta(z)=1+\sum_{n=1}^\infty \hat{x}_{n} z^{k}
\end{equation}
where 
$$
\left(\begin{array}{ll}
\delta^{*}(z) & \gamma^{*}(z) \\
\gamma(z) & \delta(z)
\end{array}\right)=: \prod_{j=1}^{k}\left(\begin{array}{cc}
1 & \alpha_{j}^{*} z^{-j} \\
\alpha_{j} z^{j} & 1
\end{array}\right)
$$
Appendix I and the analysis in \cite{Latifi-Pickrell} imply that the induced random measures on $x_n$ and $\hat{x}_n$ have identical moments. In particular, this leads to the variance identity  $E(x_n x_n^*) = E(\hat{x}_n \hat{x}_n^*)$, that reads 

\begin{equation*}
     \sum_{S\in \calA_n}\left( \prodiS \frac{1}{i\beta +1 } \right) = \left(\begin{array}{c} \beta^{-1} + n -1 \\ n \end{array}\right)
\end{equation*}
Using Newton's binomial formula we can obtain a sum which now runs over integer subsets of arbitrary size
\begin{align*}\label{eq:sum-super-telescop}
       \sum_{S\in \calA} \left( \prodiS \frac{1}{i\beta +1 } \right) x^{|S|} &=  \sum_{n=0}^{\infty} \sum_{S\in \mathcal{A}_n} \left( \prodiS \frac{1}{i\beta +1 } \right) x^n \\
       &= \sum_{n=0}^{\infty} \left(\begin{array}{c} \beta^{-1} + n -1 \\ n \end{array}\right) x^n \\
    &= \frac{1}{(1-x)^{1/\beta}}
\end{align*}

\end{proof}

\noindent 

For a real number $|x|<1$, above statement defines a probability measure on sets in $\calA$ with normalized density at $S\in \calA$ given by 

\begin{equation}\label{eq:measure-super-telescop}
S \quad \longrightarrow \quad (1-x)^{1/\beta}  \left( \prodiS \frac{1}{i\beta +1 } \right) x^{|S|} 
\end{equation}

Another way of producing measures on integer subsets comes from lattice models in statistical mechanics where the spin at each site can be either up or down, where the spin up configuration characterizes elements of the integer subset . We shall construct such  exactly solvable Ising models.  \\






 Here is the organization for the rest of the paper. In section \ref{sec:tel-cancellation} we explore the cancellation patters and prove theorem \ref{thm:poles}. In section \ref{sec:lattice} we study statistical physics lattice models that are exactly solvable as a consequence of super telescoping formula, with surprisingly simple partition functions.


\section{Telescoping Cancellations}\label{sec:tel-cancellation}

Following the discussion in section \ref{sec:suptel}, formula \ref{eq:var-formula} can be seen as a generalization of the telescoping formula in calculus. To see the cancellation patterns and to prove theorem \ref{thm:poles} we need some machinery. First, a partial fraction decomposition of the terms appearing in the sum. Second, we need a geometric visualization of integer subsets and specific families of 'reflections'. The reflections are going to pair integer subsets and provide necessary cancellations in the sum. \footnote{Part of this section has been presented and later removed from the arxiv preprint of \cite{Latifi-Pickrell}}

\subsection{Partial Fraction Formula}

Let us recall formula \ref{eq:var-formula}, 
$$
\sum_{S\in \mathcal{A}_n} \prodiS \frac{1}{i\beta +1 } = \left(\begin{array}{c} \beta^{-1} + n -1 \\ n \end{array}\right)
$$
To explore the cancellations in this infinite sum, note that for any polynomial of the form $Q(x)=(x-\alpha_1)\dots(x-\alpha_m)$, we can write a partial fraction formula as follows
 
 \[
 \frac{1}{Q(x)} = \sum_{i=1}^m \frac{1}{Q'(\alpha_i)} \frac{1}{x-\alpha_i}
  \]
In our case, for a fixed subset $S$, if we define  $Q(x) = \prod_{i \in \partial S} (x + \frac{1}{i})$ then 
\[
Q'\left(\frac{-1}{q}\right) = \prod_{i \in \partial S \backslash \{ q\}} \left(  \frac{1}{i} - \frac{1}{q} \right)
\]
and using the partial fraction formula

\[
\prod_{i \in \partial S } \frac{1}{(\beta+ \frac{1}{i} )} = \sum_{q \in \partial S} \left[  \prod_{i \in \partial S \backslash \{ q\}} \left(  \frac{1}{i} - \frac{1}{q} \right)^{-1} \right] \frac{1}{\beta + \frac{1}{q}}
\]
multiplying both sides with appropriate factors we arrive at

\[
\prod_{i \in \partial S } \frac{1}{(i \beta+ 1 )} = \sum_{q \in \partial S}  \left[   \frac{q^{2L(S)-1}}{\prod_{i \in \partial S \backslash \{ q\}} (q-i)} \right] \frac{1}{q \beta + 1}
\]
where $L(S)$ denotes the number of components of $S$. As a result, we have the following representation for our infinite sum

\begin{equation}
\sum_{S\in \mathcal{A}_n} \prodiS \frac{1}{i\beta +1 }=\sum_{S \in \calA_n} \sum_{q \in \partial S} \frac{C_q(S)}{(q\beta+1)}        
\end{equation}

\begin{equation}\label{c-q-s}
    C_q(S) = \frac{q^{2L(S)-1}}{\prod_{i \in \partial S \backslash \{ q\}} (q-i)}
\end{equation}

Now to change the order of these sums we need to restrict ourselves to the bounded region $[0,M] \subset \mathbb{R}$. Let's define $\calA_{n,M}= \{ S \in \calA_n: S \subset [0,M] \}$ and write 

\begin{equation}
\sum_{S\in \mathcal{A}_n} \prodiS \frac{1}{i\beta +1 } = \lim_{M \rightarrow \infty} \sum_{S \in \calA_{n,M}} \sum_{q \in \partial S} \frac{C_q(S)}{(q \beta+1)} = \lim_{M \rightarrow \infty} \sum_{q=0}^M \sum \limits_{\substack{%
     S \in \calA_{n,M}\\
     q \in \partial S}}          \frac{C_q(S)}{(q \beta+1)} 
\end{equation}
where the last sum is taken over all subsets $S\in \mathcal{S}_M$ having $q$ as a boundary point. At this point, we claim that for any positive integer $q$ and for sufficiently large $M$

\begin{equation}\label{q-sets-sum}
 \sum \limits_{\substack{%
     S \in \mathcal{S}_M\\
     q \in \partial S}}         C_q(S) = 0   
\end{equation}
To establish this cancellation we will need to develop some geometric machinery in what follows.   

\subsection{Integer Subset Diagrams}
An integer subset diagram with distinguished boundary point $q$, consists of the following information: 
\begin{enumerate}
    \item A horizontal axis representing integers in $[0,M]$. 
    \item A vertical axis representing special boundary integer $q$.
    \item Horizontal line segments representing a subset $S$ of $[0,M]$ with integer boundary points and $q \in \partial S$. 
    
\end{enumerate}

\begin{figure}[ht]
  \centering
  \begin{tikzpicture}
  
  \tiny
  
   \foreach \x in {0,1,...,4,5}
        \node[draw,circle,inner sep=0.5pt,fill] at (\x,5) [label=above:$\x$] {};
  
    \draw (0,5) -- (5,5);
    
    \draw (5.8,5) -- (5.85,5);
    \draw (6,5) -- (6.05,5);
    \draw (6.2,5) -- (6.25,5);
   
    \node[draw,circle,inner sep=0.5pt,fill] at (7,5) [label=above:$ M-2$] {};
    \node[draw,circle,inner sep=0.5pt,fill] at (8,5) [label=above:$M-1$] {};
    \node[draw,circle,inner sep=0.5pt,fill] at (9,5) [label=above:$M$] {};
    \draw (7,5) -- (9,5);
    
    \node[draw,circle,inner sep=0.5pt,fill] at (2,5.5) [label=above:$q$] {};
    \draw[thick] (2,2) -- (2,5.5)  {};
    
    \node at (-1.5,4) {$S= [1,2] \cup [3,4]$};
    \draw (1,4) -- (2,4); 
    \node[draw,circle,inner sep=1pt,fill] at (1,4)  {};
    \node[draw,circle,inner sep=1pt,fill] at (2,4)  {};
    \draw (3,4) -- (4,4); 
    \node[draw,circle,inner sep=1pt,fill] at (3,4)  {};
    \node[draw,circle,inner sep=1pt,fill] at (4,4)  {};
    
    \node at (-1.5,3) {$S'= [2,5] \cup [M-1,M]$};
    \draw (2,3) -- (5,3); 
    \node[draw,circle,inner sep=1pt,fill] at (2,3)  {};
    \node[draw,circle,inner sep=1pt,fill] at (5,3)  {};
    \draw (8,3) -- (9,3); 
    \node[draw,circle,inner sep=1pt,fill] at (8,3)  {};
    \node[draw,circle,inner sep=1pt,fill] at (9,3)  {};

  \end{tikzpicture}
  \caption{The diagram for the boundary point $q=2$  with two integer subsets of $[0,M]$. }
\end{figure}
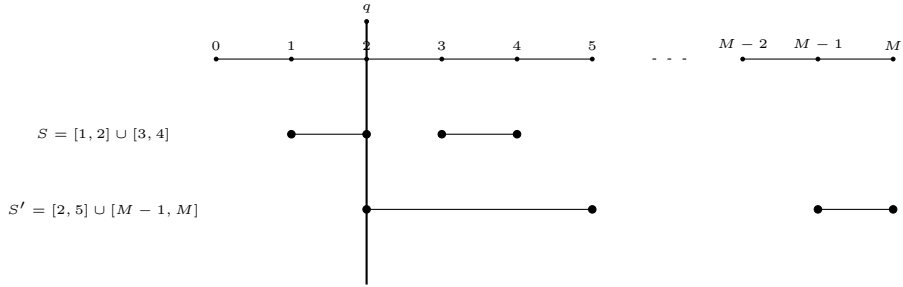

In general, any integer point exactly belongs to one of these three classes relative to $S$: 
\begin{itemize}
     \item Boundary point: $x \in \partial S$
    \item Internal point: $x \in S^\circ$
    \item External point: $x \in (S^c)^\circ$
\end{itemize}

\begin{definition}
The $d$-reflection of the set $S$ with respect to $q$, denoted by $R_d(S)$, is defined to be the symmetric reflection with respect to $q$ when the reflection only applies to the points inside the window $[q-d,q+d]$. In more precise terms, $R_d(S)$ is the closure of the following set 
\[
\{  q + a : a\in [-d,d] , \, q-a \in S \} \cup
 \{  q + a : a\in \mathbb{R} , \, |a| > d, \, q + a \in S \}    
\]
\end{definition}

The point of this definition is to come up with the appropriate values for $d$ such that the contributions from $S$ and $R_d(S)$ cancels in the sum (\ref{q-sets-sum}), i.e. 

\[
C_q(S) = -C_q(R_d(S))
\]
We use values for $d$ is such a way that $R_d(S)$ stays inside $[0,M]$ for any $S\subset [0,M]$. 

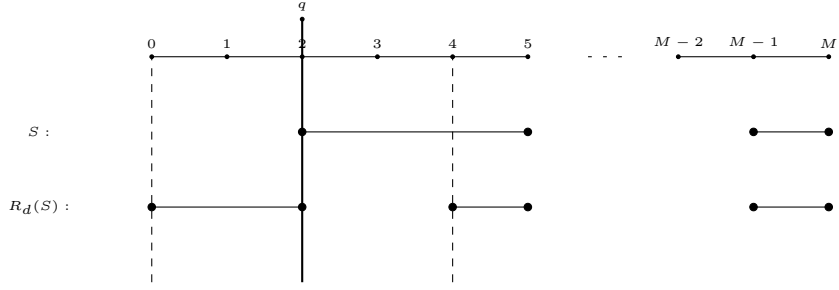
\begin{figure}[ht]
  \centering
  \begin{tikzpicture}
  
  \tiny
  
   \foreach \x in {0,1,...,4,5}
        \node[draw,circle,inner sep=0.5pt,fill] at (\x,5) [label=above:$\x$] {};
  
    \draw (0,5) -- (5,5);
    
    \draw (5.8,5) -- (5.85,5);
    \draw (6,5) -- (6.05,5);
    \draw (6.2,5) -- (6.25,5);
   
    \node[draw,circle,inner sep=0.5pt,fill] at (7,5) [label=above:$ M-2$] {};
    \node[draw,circle,inner sep=0.5pt,fill] at (8,5) [label=above:$M-1$] {};
    \node[draw,circle,inner sep=0.5pt,fill] at (9,5) [label=above:$M$] {};
    \draw (7,5) -- (9,5);
    
    \node[draw,circle,inner sep=0.5pt,fill] at (2,5.5) [label=above:$q$] {};
    \draw[thick] (2,2) -- (2,5.5)  {};
    
    \draw[dashed] (0,2) -- (0,5)  {};
    \draw[dashed] (4,2) -- (4,5)  {};

    \node at (-1.5,4) {$S:$};
    \draw (2,4) -- (5,4); 
    \node[draw,circle,inner sep=1pt,fill] at (2,4)  {};
    \node[draw,circle,inner sep=1pt,fill] at (5,4)  {};
    \draw (8,4) -- (9,4); 
    \node[draw,circle,inner sep=1pt,fill] at (8,4)  {};
    \node[draw,circle,inner sep=1pt,fill] at (9,4)  {};
    
    \node at (-1.5,3) {$R_d(S):$};
    \draw (0,3) -- (2,3); 
    \node[draw,circle,inner sep=1pt,fill] at (0,3)  {};
    \node[draw,circle,inner sep=1pt,fill] at (2,3)  {};
    \draw (4,3) -- (5,3); 
    \node[draw,circle,inner sep=1pt,fill] at (4,3)  {};
    \node[draw,circle,inner sep=1pt,fill] at (5,3)  {};
    \draw (8,3) -- (9,3); 
    \node[draw,circle,inner sep=1pt,fill] at (8,3)  {};
    \node[draw,circle,inner sep=1pt,fill] at (9,3)  {};

  \end{tikzpicture}
  \caption{An example of a $d$-reflection in diagram with $d=2$.}
\end{figure}

Note, from this example, a $d$-reflection can possibly change the number of components. However, this doesn't happen if both $q+d$ and $q-d$ are external points of $S \cup R_d(S)$ and a  $d$-reflection will preserve the number of components in that case. 

\subsection{Pair Cancellation}
A key property of the $d$-reflection is that since $q$ is a boundary point, $S$ cannot be symmetric in a small neighborhood of $q$ and hence $R_d(S) \ne S$. This fact along with the identity $R_d(R_d(S))=S$ implies that integer subsets in $\mathcal{S}_m$ come in pairs with respect to any $d$-reflection.  We want to show that the contribution to the sum vanishes in each pair, for specific values of $d$.

\begin{lemma} \label{gapreflectionlemma}
If both $S$ and $R_d(S)$ have the same number of components then 
\[
C_q(R_d(S)) = -C_q(S)
\]
and hence the corresponding pair of terms in (\ref{q-sets-sum}) cancels. 
\end{lemma}

\begin{proof}
By assumption the numerator $q^{2L(S)-1}$ in
\[
 C_q(S) = \frac{q^{2L(S)-1}}{\prod_{i \in \partial S \backslash \{ q\}} (q-i)}
\]

is the same for both $S$ and $R_d(S)$. On the other hand, a $d$-reflection will not change the magnitude of the factors $\frac{1}{q-i}$ and hence $|C_q(R_d(S))| = |C_q(S)|$. To get the sign we need to check if the number of points in $\partial S$ on the right-hand side of $q$ is odd or even. Given that $q\in \partial S$, it is easy to see geometrically that the sign changes after a $d$-reflection.

\end{proof}

For the case $q < \frac{M}{2}$, we use $d=q$ to get the following result,
\begin{lemma}
For any $q < \frac{M}{2}$ we have 
\[
 \sum \limits_{\substack{%
     S \in \mathcal{S}_M\\
     q \in \partial S}}         C_q(S) = 0   
\]
\end{lemma}
\begin{proof}
We have a pairing of subsets in $\mathcal{S}_M$ given by the $q$-reflections operator. Using lemma \ref{gapreflectionlemma} we only have to prove the cancellation when the number of components is altered after a $q$-reflection. From the diagram point of view, this can only happens if $2q$ is an internal point for $S \cup R_q(S)$. Without loss of generality we only consider the case when $2q$ is an internal point for $S$ such that a $q$-reflection adds one extra component, in other words $L(R_q(S)) = L(S) + 1 $.  For such $S$, we have
\[
C_q(R_q(S)) = \frac{q^{2L(R_q(S))-1}}{\prod_{i \in \partial R_q(S) \backslash \{ q\}} (q-i)} 
\]

Rewriting in terms of $S$ and using the fact that $0$ and $2q$ are the only boundary points in $R_q(S)$ absent in $S$, we get

\[
C_q(R_q(S)) = \frac{1}{(-q)(q)} \frac{q^{2L(S)+1}}{\prod_{i \in \partial S \backslash \{ q\}} (q-i)} = \frac{-q^{2L(S)-1}}{\prod_{i \in \partial S \backslash \{ q\}} (q-i)}
\]
which proves $C_q(R_q(S)) = -C_q(S)$.
\end{proof}

The second type of cancellation happens for $n \leq q \leq M-n$. To find the appropriate $d$ we need the following definition, 

\begin{definition}
The \textbf{Minimum Symmetric Gap} of $S$ at $q$ is the integer defined by 
\[
g(S) = \inf \{d>0 : q+d \notin S\cup R_\infty(S)\} 
\]
\end{definition}

\begin{lemma}
For any $q$, satisfying $n \leq q \leq M-n$ we have 
\[
 \sum \limits_{\substack{%
     S \in \mathcal{S}_M\\
     q \in \partial S}}         C_q(S) = 0   
\]
\end{lemma}
\begin{proof}
We pair any subset $S$ appearing in the sum with the subset $R_d(S)$ when $d=g(S)$. It is easy to see the $|S|=n$ implies $g(S) \leq n$. Having this symmetric gap will ensure that the number of components is preserved after a $d$-reflection and the proof follows from  lemma \ref{gapreflectionlemma}.
\end{proof}

To conclude, for any given integer subset one of the above cancellation statements applies and hence we only have non-zero contribution from the boundary points close to $M$. As $M \rightarrow \infty$ the analysis of poles leads to the observation that the resulting limit is a polynomial of degree $n$ in $\beta^{-1}$ (but this does not give us an explicit formula at this point), finishing the proof of theorem \ref{thm:poles}.

 \section{Lattice Models}\label{sec:lattice}

The super telescoping formula  (\ref{eq:super_tel}) can be used to find explicit expressions for  partition functions of certain lattice models. Consider the vector valued lattice model on $\mathbb{Z}^+=\{ 0,1,2,\dots\}$ where the lattice configuration at any site $i\in \mathbb{Z}^+$ is an infinite binary vector $\Vec{\sigma}(i) = (\sigma_1 (i), \sigma_2 (i), \dots)$, with $\sigma_k (i) \in \{0,1 \}$ being non-zero only in finitely many components. The temperature-dependent Hamiltonian of the model is of the form

   \begin{equation} \label{eq:vec_hamiltonian}
               H(\Vec{\sigma}) = -\sum_{i} J_i(\beta) \; \nabla \Vec{\sigma}(i) \cdot \nabla \Vec{\sigma}(i) - \sum_{i} \Vec{\mu} \cdot \Vec{\sigma}(i)
   \end{equation}
    
\noindent    
where $\Vec{\mu} = (\mu_1, \mu_2, \dots)$ is an infinite vector of reals and the $k$'th component of $\nabla \Vec{\sigma}(i)$ is defined as discrete differentiation operator on a lattice using

$$
\partial_k \sigma(i) = \sigma_k(i+1) - \sigma_k(i), \quad \quad \nabla \Vec{\sigma} = (\partial_1 \sigma, \partial_2 \sigma, \dots)
$$

Our choice of interaction rate $J_i(\beta)$ becomes independent of $\beta$ as $\beta \rightarrow 0$, with the limiting rate given by $J_i=-i$. We will prove the following result about the partition function $Z = \sum_{\Vec{\sigma}} e^{-\beta H(\Vec{\sigma})}$ for this model.

\begin{theorem} \label{thm:main}
For the lattice model on $\mathbb{Z}^+$ with the Hamiltonian given by (\ref{eq:vec_hamiltonian}), let $J_i(\beta) = \frac{-ln(1+i\beta)}{\beta}$ and $\mu_k = -ln(p_k)$ with $p_k$ being the $k$'th prime number. The  partition function is then given in terms of the Riemann zeta function by
$$
Z =\sum_{\Vec{\sigma}} e^{-\beta H(\Vec{\sigma}) }= \zeta^{\frac{1}{\beta}}(\beta)
$$
 
 \end{theorem}

 Looking component-wise at the above vectorial model and writing the Hamiltonian as a sum over $k$, we realize that the corresponding lattice system is made of an infinite family of independent sub-systems labeled by $k$. Each component corresponds to a binary lattice model on $\mathbb{Z}^+$ with the Hamiltonian
   $$
         H(\sigma) = - { \sum_i J_i(\beta) ( \partial \sigma(i))^2} - { \mu \sum_i \sigma(i)} 
    $$
Using the super telescoping formula, we show that this scalar model has a surprisingly simple partition function 

$$
Z = \frac{1}{(1-e^{\beta \mu})^{1/\beta}} 
$$


\subsection{Standard Ising Model}

A function $\sigma(i)$ on a lattice domain with values in $\{ 0,1 \} $ is usually called a lattice configuration. For a one dimensional lattice on integers define the discrete differentiation by

$$
\partial \sigma(i) = \sigma(i+1) - \sigma(i)
$$
The Hamiltonian 

\begin{equation}\label{eq:energy-ising-fluid}
    H(\sigma) = -J \sum_i ( \partial \sigma(i) )^2  -\mu \sum_i  \sigma(i) 
\end{equation}

\noindent
defines a probability distribution on the space of lattice configurations with the density given by 

$$
\frac{e^{-\beta H(\sigma)}}{Z} 
$$

\noindent
with $\beta$ being the inverse temperature and the partition function $Z = \sum_{\sigma} e^{-\beta H(\sigma)}$. The binary valued lattice model described above is closely related to the standard $\{-1,+1\}$ valued Ising model and for instance applies to fluid systems, where $\sigma(i)$ represent the presence or absence of a particle in a unit region of space. In the standard Ising model with $\sigma(i) \in \{ -1,+1\}$, the Hamiltonian is given by 
$$
H(\sigma)=-\sum_{\langle i j\rangle} J_{i j} \sigma(i) \sigma(j)-\mu \sum_{j} h_{j} \sigma(j)
$$

\noindent
it is not hard to see that the two models are equivalent after a shift in energy.

\subsection{Lattice Model on $\mathbb{Z}^+$}

In this subsection, we will show that certain lattice model in one dimension is related to the super telescoping formula.  The Hamiltonian of the model is given by

\begin{equation}
    H(\sigma) =  - \sum_i J_i(\beta) \; ( \partial \sigma(i))^2 - \mu \sum_i \sigma(i)
\end{equation}

$$
J_i(\beta) = -\frac{\ln(1+i\beta)}{\beta}\; , \quad i\in  \zplus
$$

For this model to be well-defined, $\sigma(i)\neq 0$ only for finitely many $i$'s. To arrive at our desired formula, notice that 

\begin{align}
            e^{-\beta H(\sigma)} &= \left( \prod_i e^{\beta J_i(\beta) |\partial \sigma(i)|}\right) e^{\beta \mu \sum_i |\sigma(i)| } \\
    &= \left( \prod_i \left( \frac{1}{1 + i\beta} \right)^{ |\partial \sigma(i) |} \right) e^{\beta \mu \sum_i |\sigma(i)| }
\end{align}
Here, we used the fact that $e^{\beta J_i(\beta)} = \frac{1}{1+i\beta}$. In terms of subset notation, introducing $x = e^{\beta \mu}$ the above expression is equal to

    $$  \left( \prodiS \frac{1}{1 + i\beta} \right) x^{|S|}$$

\noindent
and hence the partition function of our lattice model is given by 

$$
Z = \frac{1}{(1-x)^{1/\beta}} = \frac{1}{(1-e^{\beta \mu})^{1/\beta}} 
$$
\noindent




    
%

\subsection{Vectorial lattice model and Riemann zeta function}\label{sec:number-theory}

Historically motivated by the Hilbert-Polya conjecture, efforts have been made to construct field theories establishing correspondences between number theory and ideas in physics. Among those is the idea of primon gas discovered independently by Donald Spector \cite{Spector:1988nn} and Bernard Julia \cite{Julia:1990}. Some statistical physics models have also been constructed, for instance in \cite{Knauf:1992mi} (see references for a list of related works). 

The content of this subsection gives a proof for Theorem \ref{thm:main}. Here, we consider the lattice model given by the Hamiltonian in equation (\ref{eq:vec_hamiltonian}). Intuitively, our model corresponds to a lattice system such that each site can be occupied with a particle species labeled by a prime, with no interaction between distinct species. In the previous section, we used (\ref{eq:super_tel}) with  $x=e^{\beta \mu}$. Now, for the $k$th prime number $p_k$, we are going to assume $\mu_k =- \ln(p_k)$ which leads to $x=p_k^{-\beta}$. Formula (\ref{eq:super_tel}) then implies

\begin{equation}\label{eq:prime_tel_sum}
     \sum_{S\in \calA} \left( \prodiS \frac{1}{i\beta +1 } \right) p_k^{-\beta|S|} =\sum_{\sigma_k} e^{-\beta H(\sigma_k)} = \frac{1}{(1-p_k^{-\beta})^{1/\beta}}
\end{equation}

\noindent
we now consider an infinite family of independent lattice system, one for each prime number $p_k$. The partition function is the product of the partition function for each individual subsystem and using the Euler product formula we observe that 

\begin{equation}\label{eq:prime_prod}
  \prod_{k} \left( \sum_{\sigma_k} e^{-\beta H(\sigma_k)} \right) =\zeta^{1/\beta}(\beta)
\end{equation}
\noindent



It follows from our construction that the associated Hamiltonian is given by

\begin{equation}\label{eq:energy-ising-fluid}
               H(\Vec{\sigma}) = \sum_k \left[ -\sum_{i} J_{i}(\beta) \; (\partial \sigma_k(i))^2 - \sum_{i} \mu_k \sigma_k(i) \right]
\end{equation}
\noindent
in this model $\sigma_k(i) \in \{0,1 \}$ depending whether or not a particle of species $k$ is present at site $i$.  Let's write this in the standard form of a vector lattice model with $\Vec{\sigma}(i)$ to represent all degrees of freedom at site $i$. The energy function can then be written in a vector form 

\begin{equation}\label{eq:energy-ising-fluid}
               H(\Vec{\sigma}) = -\sum_{i}J_{i}(\beta) \; \nabla \Vec{\sigma}(i) \cdot \nabla \Vec{\sigma}(i) - \sum_{i} \Vec{\mu} \cdot \Vec{\sigma}(i)
\end{equation}
This completes the proof of theorem \ref{thm:main} and the partition function of the above model is given by $\zeta^{1/\beta}(\beta)$. We note that there are  Lee-Yang type theorems for the zeros of partition functions in terms of inverse temperature. These types of theorems can be investigated for the limiting lattice system as $\beta \rightarrow 0$.

\begin{remark}
From an alternative point of view, a lattice configuration $\Vec{\sigma}$ also defines a random function
 $$f:\mathbb{Z}^+\longrightarrow \mathbb{N}$$
$$ f(i) = \prod_k p_k^{\sigma_k(i)}$$
\noindent
hence our construction gives a way of generating random integer valued functions. 
\end{remark}


\newpage
\section*{Appendix I: Super Telescoping Formula and Verblunsky Theory}\label{sec:origin}

The following section explores the origin of the super telescoping formula including topics related to Gaussian free fields, Verblunsky coefficients and loop group factorization. It is a brief version of discussions that appeared in \cite{Latifi-Pickrell}.  The content and notation of the section is mostly independent of the rest of paper. 
\\

\textbf{Notation:} In this section, $\beta$ represent the inverse temperature parameter associated with the Gaussian free field. Complex conjugation is sometimes denoted by $*$ and $f^*(z) := \bar f (\frac{1}{\bar z}) = f(\frac{1}{z^*})^*$. We will also use the multi-index notation as follows. For a multi-index $p=(p(1), p(2), \ldots)$ we define $$p !:=\prod_{j} p(j) !$$ $$|p|:=\sum_{j} p(j)$$  $$\operatorname{deg}(p):=\sum_{j} j p(j)$$ \\

\subsection*{Verblunsky Correspondence}

The  Verblunsky correspondence associates to a probability measure on $S^1$ a sequence of complex numbers in the unit disc, $\Delta= \{z \in \mathbb{C}: |z|<1 \}$

$$\mu \leftrightarrow \{\alpha_n\} $$.

Given a probability measure $\mu$ on $S^1$, one finds a set of monic orthogonal polynomials  $\{ p_n\} $ relative to $\mu$, satisfying the recursive relation 
$$
p_{n}(z)=z p_{n-1}(z)+\alpha_{n}^{*} z^{n-1} p_{n-1}^*(z), \quad n>0
$$

with $\{\alpha_n\}$ being the unique Verblunsky sequence associated to $\mu$. Conversely, given a sequence $\{\alpha_n\}$ with $\alpha_n \in \Delta$, the measure $\mu$ is given by 

$$
d \mu=\lim _{N \rightarrow \infty} \frac{\prod_{n=1}^{N}\left(1-\left|\alpha_{n}\right|^{2}\right)}{\left|p_{N}(z)\right|^{2}} \frac{d \theta}{2 \pi}
$$

\subsection*{Intermediate Sequence}

For a continuous measure $d\mu = e^{f} \, \frac{d\theta}{2 \pi}$ with $ f(z) = \sum_n f_n z^n$, the well-known Szego formula says 

$$
\exp \left(-\sum_{n=1}^{\infty} n\left|f_{n}\right|^{2}\right)=\prod_{n=1}^{\infty}\left(1-\left|\alpha_{n}\right|^{2}\right)^{n}
$$

Any sequence $\{ f_n\} $ is then associated to a Verblunsky sequence $\{ \alpha_n \}$. To arrive at the super telescoping formula we need to introduce an intermediate sequence $\{ x_n \} $ which can be obtained from both $\{ f_n\}$ and $\{ \alpha_n \}$. In terms of $\{ f_n\}$,  lets define $\{ x_n \} $ by

$$
e^{-f}=1+\sum_{1}^{\infty} x_{n} z^{n}
$$

$$
f(z) = \sum_{n=1}^{\infty} f_n z^n
$$

Coming from the other direction, let us now construct  $\{ \hat{x}_n\}$ from the Verblunsky sequence $\{ \alpha_n \} $.  Instead of conventional orthogonal polynomials $p_n(z)$ it seems more natural for us to use the reversed polynomials $r_n(z) = z^n p_n^* (z)$, where the Szego recursion can be written as 
$$
\left(\begin{array}{c}
r_{n}^{*} \\
r_{n}
\end{array}\right)=\left(\begin{array}{cc}
1 & \alpha_{n}^{*} z^{-n} \\
\alpha_{n} z^{n} & 1
\end{array}\right)\left(\begin{array}{c}
r_{n-1}^{*} \\
r_{n-1}
\end{array}\right)=\prod_{k=1}^{n}\left(\begin{array}{cc}
1 & \alpha_{k}^{*} z^{-k} \\
\alpha_{k} z^{k} & 1
\end{array}\right)\left(\begin{array}{l}
1 \\
1
\end{array}\right)
$$

The above expression represent a loop group factorization for a loop with values in $SU(1,1)$.  The sequence $\{ \hat{x}_n\}$ can then be obtained from $\{ \alpha_n \}$ through 

\begin{equation} \label{eq:def_xhat}
    \gamma_{2}+\delta_{2}:=1+\sum_{k=1}^n \hat{x}_{k} z^{k}
\end{equation}

where 
$$
\left(\begin{array}{ll}
\delta_{2}^{*} & \gamma_{2}^{*} \\
\gamma_{2} & \delta_{2}
\end{array}\right)=: \prod_{j=1}^{k}\left(\begin{array}{cc}
1 & \alpha_{j}^{*} z^{-j} \\
\alpha_{j} z^{j} & 1
\end{array}\right)
$$

The analysis in \cite{Latifi-Pickrell} suggests that given certain distribution of $\{ f_n\} $ and $\{ \alpha_n \}$,   the two sequences $\{ x_n\}$ and $\{ \hat{x}_n\}$ should have identical moments. 

\subsection*{GFF and Super-telescoping formula}

The Gaussian free field can be described as a product measure on $\{ f_n\}$ in terms of inverse temperature parameter $\beta$ 

\begin{equation} \label{eq:eq:dnu_dist}
d\nu_{\beta} = \prod_{n=1}^{\infty} \frac{n \beta}{\pi} e^{-n \beta\left|f_{n}\right|^{2}} d \lambda\left(f_{n}\right)
\end{equation}

A reasonable question here is to find the measure on the space of Verblunsky sequences $\{\alpha_n\}$ associated with this Gaussian. The follwoing answer has been given in \cite{Chhaibi} and \cite{Latifi-Pickrell} 

\begin{equation} \label{eq:alpha_dist}
\prod_{n=1}^{\infty} \frac{n \beta}{\pi}\left(1-\left|\alpha_{n}\right|^{2}\right)^{n \beta-1} d \lambda\left(\alpha_{n}\right)
\end{equation}

The basic idea of \cite{Latifi-Pickrell} is that $\{ \hat{x}_n\}$ and $\{ x_n\}$ defined in terms of $\{\alpha_n\}$ and $\{ x_n\}$ respectively, have identical moments with respect to this measures. The rest of this section is to establish this correspondence and show that it leads to the super telescoping formula. As a complex Gaussian  random variable, $f_n$ satisfies $E\left(\left|f_{n}\right|^{2 k}\right)=\frac{k !}{(n \beta)^{k}}$ and therefore 

\begin{equation} \label{eq:fn_moments}
E\left(\prod_{n \geq 1} f_{n}^{p(n)}\left(\prod_{m \geq 1} f_{m}^{q(m)}\right)^{*}\right)=\prod_{n \geq 1} \frac{p(n) !}{(n \beta)^{p(n)}}
\end{equation}

For $f(z) = \sum_{n=1}^\infty f_n z^n$, with $\{ f_n \} $ distributed according to the measure $d\nu_{\beta}$, let $\{ x_n\} $ and $\{ y_n\} $ be random sequences defined by 

$$
e^{-f} = 1 + \sum_{n=1}^\infty x_n z^n
$$

$$
e^{f} = 1 + \sum_{n=1}^\infty y_n z^n
$$

$\{ x_n\} $ and $\{ y_n\} $ have identical distributions induced from $d\nu_{\beta}$ but $\{ y_n\} $ is more convenient to work with. To write $y_n$ in terms of $f_n$ we observe that

$$
 e^f = 1 + (\sum_1^\infty f_n z^n) + \frac{1}{2!} (\sum_1^\infty f_n z^n)^2 + \dots 
$$

which in terms of multi-indices implies  $y_{n}=\sum_{\{J: d e g(J)=n\}} \frac{1}{J !} f^{J}$. Now, using (\ref{eq:fn_moments}), the moments of $\{ y_n\} $ and hence $\{ x_n\} $, are given by

\begin{equation} \label{eq:xn_moments}
\begin{split}
     E\left(x^{p}\left(x^{q}\right)^{*}\right) &= E\left(\prod_{n \geq 0}\left(\sum_{J_{n}} \frac{1}{J_{n} !} f^{J_{n}}\right)^{p(n)}\left(\prod_{m \geq 0}\left(\sum_{K_{m}} \frac{1}{K_{m} !} f^{K_{m}}\right)^{q(m)}\right)^{*}\right) \\
    &= \sum_{J_{n, r}, K_{m, s}} \frac{1}{\left(\prod_{n, r} J_{n, r} !\right)\left(\prod_{m, s} K_{m, s} !\right)} E\left(f^{\sum_{n, r} J_{n, r}}\left(f^{\sum_{m, s} K_{m, s}}\right)^{*}\right)\\ 
    &= \sum \frac{1}{\prod_{n, r}\left(J_{n, r} !\right) \prod_{m, s}\left(K_{m, s} !\right)} \frac{\left(\sum_{n, r} J_{n, r}\right) !}{\prod_{u} u^{\left(\sum_{n, r} J_{n, r}\right)(u)}} \beta^{-\left|\sum_{n, r} J_{n, r}\right|}
\end{split}
\end{equation}

here the sum is over all multi-indices $J_{n, r}, K_{m, s}$ satisfying  the constraints 
$$ \operatorname{deg}\left(J_{n, r}\right)=n$$
$$\operatorname{deg}\left(K_{m, s}\right)=m$$
$$\sum_{n, r} J_{n, r}=\sum_{m, s} K_{m, s}$$
 and for given $n, m>0$ we have  $1 \leq r \leq p(n)$ and $ 1 \leq s \leq q(m)$. Equation in (\ref{eq:xn_moments}) can be written in a polynomial form as 
 
\begin{equation} \label{eq:xn_moments_short}
 E\left(x^{p}\left(x^{q}\right)^{*}\right) = \sum_{k=1}^{\operatorname{deg}(p)} a(p, q, k) \beta^{-k} 
\end{equation}
$$
a(p, q, k)=\sum \frac{1}{\prod_{n, r}\left(J_{n, r} !\right) \prod_{m, s}\left(K_{m, s} !\right)} \frac{\left(\sum_{n, r} J_{n, r}\right) !}{\prod_{u} u^{\left(\sum_{n, r} J_{n, r}\right)(u)}}
$$

A careful look at (\ref{eq:xn_moments}) and the computation behind it leads to the following special case for $E(x_n x_n^*)$ (see section 5.2 of \cite{Latifi-Pickrell} for a detailed argument)
 
  $$ E(x_n x_n^*)  =  \left(\begin{array}{c} \beta^{-1} + n -1 \\ n \end{array}\right)$$

Let us now start from the distribution of $\{ \alpha_n \} $ given in (\ref{eq:alpha_dist}) and compute the corresponding distribution for $\{ \hat{x}_n \} $.  Separating the angular and radial components of $\alpha_n$, for a positive integer $k$ we have

$$
E\left(\left|\alpha_{n}\right|^{2 K}\right)=2 \pi \int_{r=0}^{1} r^{2 K} \frac{n \beta}{\pi}\left(1-r^{2}\right)^{n \beta-1} r d r
$$

Using the definition of the beta function $
\mathrm{B}(x, y)=\int_{0}^{1} t^{x-1}(1-t)^{y-1} d t
$ we can conclude

\begin{equation} \label{eq:alpha_moments}
\begin{split}
     E\left(\left|\alpha_{n}\right|^{2 K}\right) &= n \beta B(K+1, n \beta)\\
     &= \frac{\Gamma(K+1) \Gamma(n \beta+1)}{\Gamma(n \beta+K+1)} \\
    &= \frac{K !}{(n \beta+1) \ldots(n \beta+K)}
\end{split}
\end{equation}

Above equation leads to the following moment relation. For multi-indices $p$ and $q$ we have 

$$
E\left(\alpha^{p}\left(\alpha^{q}\right)^{*}\right)=\prod_{n = 1}^\infty \frac{p(n) !}{(n \beta+1) \ldots(n \beta+p(n))}
$$

 Combining this with (\ref{eq:def_xhat}) one can obtain
 
 \begin{equation} \label{eq:alpha_moments}
\begin{split}
     E(\hat{x}_n \hat{x}_n^*)  &= \sum_{S\in \mathcal{A}_n} \prod_{i\in \partial S} E(|\alpha_i|^2) \\
     &= \sum_{S\in \mathcal{A}_n} \prod_{i\in \partial S} \frac{1}{i \beta +1 }
\end{split}
\end{equation}

We have shown the following: 

\begin{proposition}
The variance identity  $E(x_n x_n^*) = E(\hat{x}_n \hat{x}_n^*)$ is equivalent to the super-telescoping identity given by  \begin{equation}
     \sum_{S\in \mathcal{A}_n} \prod_{i\in \partial S} \frac{1}{i \beta +1 } = \left(\begin{array}{c} \beta^{-1} + n -1 \\ n \end{array}\right) 
\end{equation}
\end{proposition}

\subsection*{Remark on Generating functions}

Combinatorially, the expression $(\sum_{i=0}^\infty a_i x^i)^n = \sum_{j=0}^\infty c_j x^j $ can be used to calculate a sum that is indexed by different ways of grouping $j$ items in $n$ groups. We can ask whether or not such a generating function exists that sums over all integer subsets of $\mathbb{Z}^+$. Recall that summing over $n$, the super telescoping formula corresponds to the $z^n$ coefficient in the expansion

$$ \sum_{S\in \mathcal{A}} \left[ \prod_{i\in \partial S} \frac{1}{i \beta +1 }\right] z^{|S|} = \frac{1}{(1-z)^{1/\beta}}$$

This suggest a form of generating function for the right hand side. However, to our knowledge, there exist only a matrix analog of generating function given by 
    
    $$
    \begin{bmatrix} 
    A(x) & B(x) \\
    B(1/x) & A(1/x)
    \end{bmatrix} = 
    \dots
    \begin{bmatrix}
    1 & \eta_2 x^{-2} \\
    \eta_2 x^{2} & 1
    \end{bmatrix}
    \begin{bmatrix}
    1 & \eta_1 x^{-1} \\
    \eta_1 x^1 & 1
    \end{bmatrix}
    \begin{bmatrix}
    1 & \eta_0  \\
    \eta_0  & 1
    \end{bmatrix}
    $$
    
    $$
    \eta_i = \frac{1}{i \beta+1}
    $$

    This product is real valued version of the loop group factorization in $SU(1,1)$. Doing the infinite matrix product give rise to a calculation involving super telescoping formula and one can conclude
    
    $$
    A(x) = \frac{1}{(1-x)^{1/\beta}}
    $$
    \\


\begin{thebibliography}{9}

\bibitem[1]{Latifi-Pickrell} Javad Latifi, M., Pickrell, D.\ 2020.\ The Exponential of the $S^1$ Trace of the Free Field and Verblunsky Coefficients, Rocky Mountain Journal of Mathematics, Vol. 52, 2022




\bibitem[2]{Spector:1988nn}
D.~Spector,
Commun. Math. Phys. \textbf{127} (1990), 239
doi:10.1007/BF02096755


\bibitem[3]{Julia:1990} Bernard L. Julia, Statistical theory of numbers, in Number Theory and Physics, eds. J. M. Luck, P. Moussa, and M. Waldschmidt, Springer Proceedings in Physics, Vol. 47, Springer-Verlag, Berlin, 1990, pp. 276–293

\bibitem[4]{Chhaibi}R. Chhaibi and J. Najnudel,On the circle, $GMC^\gamma = \varprojlim C\beta E\_n$ for $\gamma = \sqrt\{\frac\{2\}\{\beta\}\}, $ $( \gamma \leq 1 )$  , arxiv: 1904.00578, 2019

\bibitem[5]{Knauf:1992mi}
A.~Knauf,
Commun. Math. Phys. \textbf{153} (1993), 77-116
doi:10.1007/BF02099041


\bibitem[6]{Schumayer_2011} Schumayer, Dániel and Hutchinson, David A. W, Physics of the Riemann hypothesis, Rev. Mod. Phys. 83, 307, 2011


\bibitem[7]{Knauf_1998} Andreas Knauf, Number theory, dynamical systems and statistical mechanics, Reviews in Mathematical Physics V. 11, 1999

\bibitem[8]{Ruelle} Ruelle D. Some remarks on the location of zeroes of the partition function for lattice systems. Communications in Mathematical Physics, 31, 265-277, 1973

\bibitem[9]{Newman_1974} Newman C. Zeros of the partition function for generalized ising systems. Communications on Pure and Applied Mathematics, 27, 143-159, 1974
  






\end{thebibliography}
\end{document}